\documentclass[a4paper,11pt]{article}

\usepackage{graphicx}
\usepackage{a4wide}
\usepackage[T1]{fontenc}
\usepackage[utf8]{inputenc}
\usepackage{amsmath, amssymb, amsthm}
\usepackage[english]{babel}
\usepackage{eurosym}
\usepackage{tikz}
\usepackage{verbatim,refcount}
\usepackage{fancybox,soul,color}
\usepackage{url}
\usepackage{hyperref}

\usepackage{subfig}
\usepackage{enumerate}
\usepackage{epsfig}
\usepackage[active]{srcltx}
\usepackage{amsfonts}
\usepackage{amsmath}
\usepackage[mathlines]{lineno}
\usepackage{xspace}
\usepackage{tikz}
\usepackage{multirow,tabu}
\usepackage{slashed}
\usepackage{eucal}
\usepackage{thmtools}
\usepackage{thm-restate}
\usepackage{hyperref}
\usepackage{textpos}
\usepackage{booktabs}

\usepackage[linesnumbered,ruled,vlined]{algorithm2e}

\usetikzlibrary{arrows,shapes,automata,backgrounds,petri}

\DeclareMathOperator{\mad}{mad}

\DeclareMathOperator{\ad}{ad}
\DeclareMathOperator{\ch}{ch}

\newtheorem{theorem}{Theorem}
\newtheorem{lemma}[theorem]{Lemma}
\newtheorem{claim}{Claim}

\def\claimb{\vcenter\bgroup\advance\hsize by -8em\noindent
\refstepcounter{claimb}\ignorespaces\it}
\makeatletter
\def\endclaimb{\rm\egroup\leqno(\theclaim)\global\@ignoretrue}
\makeatother

\newenvironment{claimproof}[1][]%
    {\noindent \emph{Proof.} {}{#1}{}}{\hfill
    $\Diamond$\vspace{1em}}

\newcounter{rulecnt}
\newcounter{confcnt}
\newcommand{\anonymous}[1]{#1}

\def\ch{\text{ch}}

\def\dD{d^\Delta}
\def\free{\bar{\pi}}

\newcommand{\SmallPicture}[2]{%
  \includegraphics[height=#2]{#1}%
  \xspace
}
\newcommand{\dice}{\smash{\SmallPicture{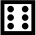}{0.8em}}}

\renewcommand{\Pr}[1]{\operatorname{Pr}\mathopen{}\left( #1 \right)\mathclose{}}
\newcommand{\Expect}[1]{\operatorname*{E}\mathopen{}\left( #1 \right)\mathclose{}}

\begin{document}

\title{Edge-coloring sparse graphs with $\Delta$ colors in quasilinear time}

\anonymous{
\date{\today}

\author{{\L}ukasz Kowalik\thanks{Institute of Informatics, University of Warsaw, Poland (\texttt{kowalik@mimuw.edu.pl}). This work is a part of project BOBR that has received funding from the European Research Council (ERC) under the European Union’s Horizon 2020 research and innovation programme (grant agreement No. 948057).}}
}

\maketitle

\anonymous{
\begin{textblock}{20}(-1.8, 8.3)
	\includegraphics[width=40px]{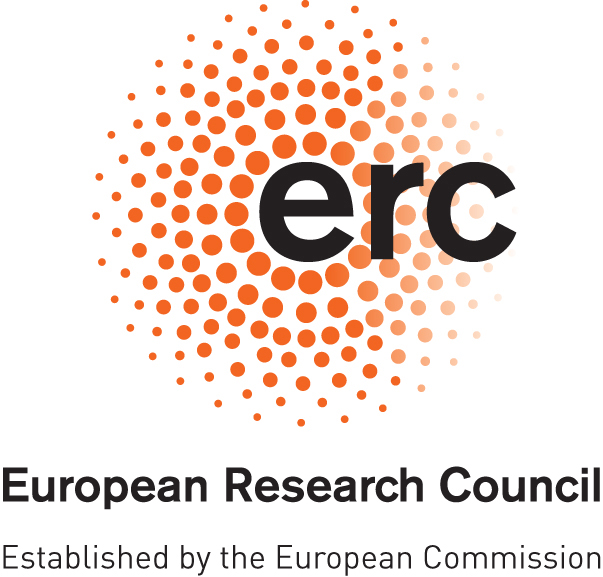}%
\end{textblock}
\begin{textblock}{20}(-2.05, 8.6)
	\includegraphics[width=60px]{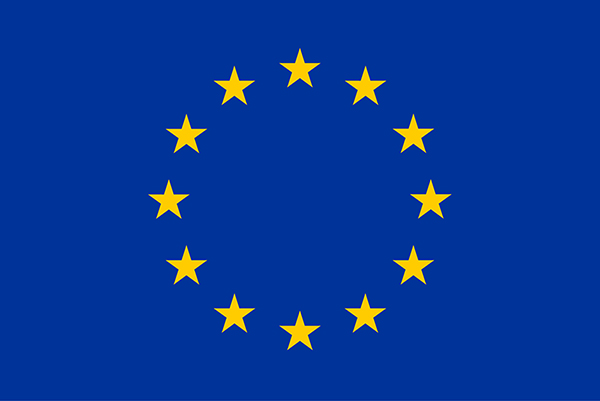}%
\end{textblock}
}

\begin{abstract}
In this paper we show that every graph $G$ of bounded maximum average degree $\mad(G)$ and with maximum degree $\Delta$ can be edge-colored using the optimal number of $\Delta$ colors in quasilinear time, whenever $\Delta\ge 2\mad(G)$. 
The maximum average degree is within a multiplicative constant of other popular graph sparsity parameters like arboricity, degeneracy or maximum density.
Our algorithm extends previous results of Chrobak and Nishizeki~\cite{chrobak-nishizeki} and Bhattacharya, Costa, Panski and Solomon~\cite{Bhattacharya23}.
\end{abstract}

\section{Introduction}
Algorithms for edge-coloring is a classic research topic which has become active again recently.
Let us recall that edge-coloring of graph $G$ is a function $\pi:E(G)\rightarrow\mathbb{N}$ which assigns different values (called {\em colors}) to incident edges. 
By {\em $k$-edge-coloring} we mean an edge coloring which uses at most $k$ colors.
The minimum number of colors that suffice to color graph $G$ is called the {\em chromatic index} of $G$ and denoted $\chi'(G)$. For a maximum degree of $G$ denoted $\Delta(G)$ it is clear that $\chi'(G)\ge\Delta(G)$, while the classic theorem of Vizing states that $\chi'(G)\le\Delta(G)+1$. 
For general graphs, determining whether $\chi'(G)=\Delta(G)$ is NP-complete, as shown by Holyer~\cite{holyer}.

\begin{table}[t]
    {
    \begin{center}
    \begin{tabular}{c|c|c|c}
        \toprule
        Number                       & \multirow{2}{*}{Graph class} & \multirow{2}{*}{Time} & \multirow{2}{*}{Reference}\\
        of colors& & &\\
        \midrule
        $(1+\epsilon)\Delta$ & general & $O(m\epsilon^{-1}\log n)$ &  Elkin and Khuzman~\cite{elkin2024deterministic} \\
        $(1+\epsilon)\Delta$ & general & $O(m\log(1/\epsilon))$ \dice  &  Assadi~\cite{assadi24} \\
        $\Delta+O(\log n)$ & general & $O(m\log\Delta)$ \dice &  Assadi~\cite{assadi24} \\
        $\Delta+2\alpha-2$ & arboricity $\alpha$ & $O(m\log\Delta)$ &  Christiansen et al.~\cite{AChristiansen-sparse-dynamic} \\
        $\Delta+1$ & general & $O(m\sqrt{n})$ &  Sinnamon~\cite{sinnamon} \\
        $\Delta+1$ & general & $O(mn^{1/3})$ \dice &  Bhattacharya et al.~\cite{Bhattacharya24} \\
        $\Delta+1$ & general & $O(m\Delta\log n)$ &  Gabow et al.~\cite{gabow} \\
        $\Delta+1$ & general & $O(m\Delta^{18})$ \dice &  Bernshteyn and Dhawan~\cite{BernshteynDhawanBoundedDelta} \\
        $\Delta+1$ & general & $O(n^2\log n)$ \dice &  Assadi~\cite{assadi24} \\
        $\Delta+1$ & arboricity $\alpha$ & $O(m\alpha\log n)$ \dice &  Bhattacharya et al.~\cite{Bhattacharya23} \\
        $\Delta+1$ & arboricity $\alpha$ & $O(m\alpha^7\log n)$ &  {\bf This work} \\ 
        $\Delta+1$ & arboricity $\alpha$ & $\tilde{O}(m\sqrt{n}\,\tfrac{\alpha}{\Delta})$ \dice &  Bhattacharya et al.~\cite{Bhattacharya23} \\
        $\chi'(G)$ & bounded treewidth& $O(n)$ &  Zhou et al.~\cite{ZhouTreewidth} \\
        $\Delta$ & bipartite & $O(m\log \Delta)$ &  Cole, Ost and Shirra~\cite{ColeOstSchirra} \\
        $\Delta$ & bounded genus, $\Delta\ge 19$ & $O(n)$ &  Chrobak and Yung~\cite{ChrobakY89} \\
        $\Delta$ & bounded genus, $\Delta\ge 9$ & $O(n\log n)$ &  Chrobak and Nishizeki~\cite{chrobak-nishizeki} \\
        $\Delta$ & planar, $\Delta\ge 9$ & $O(n)$ &  Cole and Kowalik~\cite{cole-kowalik} \\
        $\Delta$ & $\Delta\ge 2\mad(G)$ & $O(m\alpha^3\log n)$ \dice &  {\bf This work} \\ 
        $\Delta$ & $\Delta\ge 2\mad(G)$ & $O(m\alpha^7\log n)$ &  {\bf This work} \\ \bottomrule
    \end{tabular}
    \caption{\label{table:results}Summary of state-of-the-art edge-coloring algorithms. Note that $\alpha=\Theta(\mad(G))$. The dice symbol \dice denotes a randomized algorithm.}
    \end{center}
    }
\end{table}

\subsection{Algorithms for general graphs}

Vizing's proof can be easily transformed to an algorithm for edge-coloring graphs in $\Delta+1$ colors running in time $O(nm)$, where $n$ and $m$ denote the number of vertices and edges, respectively, of the input graph (we use this notation throughout the whole paper). Gabow, Nishizeki, Kariv, Leven, and Terada~\cite{gabow} were the first to show an algorithmic progress: they designed two algorithms, one running in time $O(m\sqrt{n\log n})$  and another in time $O(m\Delta\log n)$. 
Both these algorithms saw improvements recently: Sinnamon~\cite{sinnamon} obtained deterministic time $O(m\sqrt{n})$, Bhattacharya, Carmon, Costa, Solomon and Zhang~\cite{Bhattacharya24} have randomized time $O(mn^{1/3})$, while Bernshteyn and Dhawan~\cite{BernshteynDhawanBoundedDelta} randomized time $\Delta^{O(1)}m$. 
Another direction of research is improving the running time at the price of an increased number of colors: very recently, Elkin and Khuzman~\cite{elkin2024deterministic} presented a deterministic algorithm that uses at most $(1+\epsilon)\Delta$ colors for any $\epsilon\ge\frac1\Delta$ and works in time $O(m\epsilon^{-1}\log n)$, while Assadi~\cite{assadi24} was able to get the improved time of  $O(m\log(1/\epsilon))$, at the price of randomization;
see also Duan et al.~\cite{DuanHZ19} for an earlier result of this kind.

\subsection{Algorithms for graph classes}

Restricting the input graph may allow for algorithms with improved performance.
By the classical K\H{o}nig's Theorem, for bipartite graphs only $\Delta$ colors suffice, and Cole, Ost and Shirra~\cite{ColeOstSchirra} provided an $O(m\log\Delta)$-time algorithm. Zhou, Nakano and Nishizeki~\cite{ZhouTreewidth} gave an algorithm that uses $\chi'(G)$ colors and runs in linear time for graphs of bounded treewidth. Chrobak and Yung~\cite{ChrobakY89} gave a linear time algorithm that uses only $\max\{\Delta,19\}$ colors and works for graphs of bounded genus. 
Graphs from this class can be also colored using only $\max\{\Delta,9\}$ colors by an algorithm of Chrobak and Nishizeki~\cite{chrobak-nishizeki} that runs in time $O(n\log n)$. 
For the special case of planar graphs, Cole and Kowalik~\cite{cole-kowalik} improved that to linear time. 

\subsection{Uniformly sparse graphs and our result}

{\em Maximum average degree} is defined as $\mad(G)=\max_H{\rm ad}(H)$, where the maximum is over all nonempty subgraphs of $G$ and ${\rm ad}(H)=\tfrac{1}s\sum_{v\in V(H)}d_H(v)$ is the average degree of $H$. 
Hence, when $\mad(G)$ is low, graph $G$ is uniformly sparse. 
It is easy to see that $\mad(G)$ equals exactly twice the {\em maximum density} of $G$, i.e., $\max_S|E(S)|/|V(S)|$, which is close to the {\em arboricity} defined as $\alpha(G)=\max_{S, |V(S)|>1}\lceil|E(S)|/(|V(S)|-1)\rceil$. Indeed, all these three parameters (and also the degeneracy), are within a small constant of each other (see e.g.~\cite{KowalikApproxDensity}). 
In particular, when one of them is bounded by a constant, so are all the others. Moreover, planarity, bounded treewidth or bounded genus all imply bounded maximum average degree, arboricity, etc. 

Very recently, Bhattacharya, Costa, Panski and Solomon~\cite{Bhattacharya23} obtained two algorithms that use $\Delta+1$ colors and work in randomized time $O(m\sqrt{n}\Delta/\alpha(G))$ and $O(m\alpha(G)\log n)$. In particular, the latter one is quasilinear for bounded arboricity, and both algorithms improve over the results of Gabow et al.~\cite{gabow} when $\Delta(G)=\omega(\alpha(G))$. Another recent result is an algorithm of Christiansen, Rotenberg and Vlieghe~\cite{AChristiansen-sparse-dynamic} that uses $\Delta+2\alpha(G)-2$ colors and runs in deterministic time $O(m\log\Delta)$.

In this work we continue this line of research with the goal of {\em improving the number of colors used}. 
Note that the quasilinear algorithm of Chrobak and Nishizeki~\cite{chrobak-nishizeki} uses only $\Delta$ colors for the special case of graph $G$ of bounded genus $g$ and such that $\Delta\ge 9$. Since then $\mad(G) \le 6+O(g/n)$ it is natural to ask if it is possible to generalize it to any class of uniformly sparse graphs. We answer this question in the affirmative by showing the following main result.

\begin{restatable}{theorem}{ourmainthm}
    \label{thm:main}
    Every graph $G$ with $n$ vertices and $m$ edges such that $\Delta(G)\ge 2\mad(G)$ can be $\Delta(G)$-edge-colored 
    \begin{enumerate}
        \item by a randomized algorithm running in time $O(m\mad(G)^3\log n)$ in the expectation and with high probability,
        \item by a deterministic algorithm in time $O(m\mad(G)^7\log n)$.
    \end{enumerate}
\end{restatable}

In particular, whenever $\mad(G)=\log^{O(1)}n$, the algorithms work in quasilinear time. Here, by {\em high probability} we mean probability at least $1-n^{-k}$, for an arbitrary constant $k>0$. Note that as a corollary we get also the first {\em deterministic} quasi-linear time algorithm for $(\Delta+1)$-edge-coloring graphs of bounded $\mad(G)$ (or arboricity, etc.), thus derandomizing the result of Bhattacharya et al.~\cite{Bhattacharya23}. Indeed, if $\Delta(G)\ge 2\mad(G)$ we apply our result, and otherwise just apply the $O(m\Delta\log n)$-time algorithm of Gabow et al.~\cite{gabow}.

Let us comment on the assumption $\Delta(G)\ge 2\mad(G)$. Clearly, some kind of assumption is needed because there are sparse graphs of chromatic index $\Delta(G)+1$. Vizing~\cite{vizing-unsolved} (see also~\cite{stiebitz}) conjectured\footnote{The original conjecture states that in every critial graph (in a sense, a minimal graph with chromatic index $\Delta+1$) we have $|E(G)|\ge\tfrac12((\Delta(G)-1)|V(G)|+3)$.} that for $\chi'(G)=\Delta(G)$ it suffices that $\mad(G)\le\Delta-1$. While this conjecture is still open, there has been a substantial progress. Fiorini~\cite{fiorini}, Haile~\cite{haile} and Sanders and Zhao~\cite{sanders-zhao-size} obtained sufficient conditions of the form $\mad(G) \le \tfrac12\Delta + o(\Delta)$.
Next, it was improved to $\mad(G) < \tfrac23(\Delta + 1)$ by Woodall~\cite{woodall-2007,woodall-2007-erratum}, and further to $\mad(G) <\min\{\tfrac34\Delta-2, 0.738\Delta-1.153\}$ by Cao, Chen, Jiang, Liu and Lu~\cite{cao-critial-density}. All these results are in fact constructive and correspond to polynomial-time algorithms (but {\em not} quasilinear\footnote{The algorithm corresponding to the proof of Sanders and Zhao~\cite{sanders-zhao-size} can be implemented in time $O(nm)$; for the other the worst time complexity is at least that large.}). 
However, while the proof of Sanders and Zhao~\cite{sanders-zhao-size}  relies on a simple procedure for extending a partial coloring called Vizing Adjacency Lemma (VAL, see Theorem~\ref{th:VAL} below), the later proofs result in considerably more complicated algorithms. 
In fact we show that the constant 2 is optimal if we restrict ourselves to algorithms that color the graph edge by edge using VAL (see Theorem~\ref{thm:best} for a formal statement).
Since we prefer to keep this article short, we leave the task of improving the assumption as an open problem.

\subsection{Sketch of our approach}

Let $D$ be the number of available colors. 
Let us recall the following classical result (usually stated in terms of critical graphs).
Here, for a statement $P$ the expression $[P]$ equals $1$ if $P$ holds and $0$ otherwise.

\begin{theorem}[Vizing Adjacency Lemma, VAL~\cite{vizing:critical}]
    \label{th:VAL}
    Let $G$ be a simple graph and let $e=xy$ be an edge such that $x$ has at most $D - d(y) + [d(y)=D]$ neighbors of degree $D$.
    Then any partial $D$-edge-coloring of $G$ which colors a subset $E_c$ of edges of $G$, $e\not\in E_c$, can be extended to a partial $D$-edge-coloring that colors $E_c\cup\{e\}$.
\end{theorem}

In this work, an edge $xy$ that satisfies the assumption of VAL is called {\em $(D,x)$-weak}, or simply {\em $x$-weak} when $D$ is known.
Moreover, $xy$ is {\em weak} if it is $x$-weak or $y$-weak.

Note that for $D=\Delta(G)+1$, every edge is weak, just because then every vertex has no neighbors of degree $D$.
Hence, in a simple implementation of standard proofs of Vizing Theorem, we color the input graph $G$ edge by edge, using the algorithm originating from the Vizing's proof of VAL. 
This algorithm uses so-called fans and alternating paths.
Given an edge $e=xy$, an $e$-fan is roughly a certain sequence of neighbors of $x$ (equivalently, incident edges) --- a precise definition will be given later.
An {\em alternating path} is just a maximal path with edges colored alternately in two colors.
Of course the algorithm behind VAL may need to recolor some of the already colored edges. 
However, it turns out that it always suffices to recolor some edges of an $xy$-fan and a single alternating path $P_{xy}$.

When we aim at a quasilinear algorithm, even for $D=\Delta(G)+1$, the approach above poses the following problem: even if coloring edge $xy$ takes just $O(1)$ time per a recolored edge, our bound on the number of those edges is $\deg(x)+|P_{xy}|$, and the sum $\sum_{xy\in E(G)}(\deg(x)+|P_{xy}|)$ can be as large as $\Theta(nm)$, for $n:=|V(G)|$, $m:=|E(G)|$. The problem with long alternating paths has been resolved in an elegant way by Sinnamon~\cite{sinnamon}: he shows, roughly, that if we pick a {\em random} uncolored edge $xy$ out of all $\ell$ uncolored edges, then the probability that a fixed alternating path is chosen is $O(\tfrac1\ell)$. Then he argues that the {\em total} length of alternating paths is only $O(\Delta m)$, so $\Expect{|P_{xy}|}=O(\tfrac1\ell\Delta m)$. 
This means that when $\ell$ is large, i.e., most of the time, the path is usually short, and the total expected length of the alternating paths used by the algorithm is $O(H_\ell\Delta m)$, where $H_\ell=O(\log\ell)$ is the $\ell$-th harmonic number. 

Another solution for the same problem was used by Chrobak and Nishizeki~\cite{chrobak-nishizeki}. They also need a large number $\ell$ of uncolored edges.
They show that then we can pick a set $I$ of size $\Omega(\ell/\Delta(G)^4)$ of these edges so that all the corresponding alternating paths are of the same {\em type} (alternate the same pair of colors) and moreover coloring one of edges from $I$ using VAL does not interact with coloring the others, intuitively: the VAL algorithm for one of these edges works the same, independent of whether it was invoked for some other edges from $I$ before or not. Hence we can apply VAL to all the edges of $I$. The total length of the alternating paths is then linear, since they all have the same type and hence they are disjoint. A simple calculation shows that after iterating this $O(\log n\,\Delta(G)^4)$ times all edges get colored, and the total running time of coloring the $\ell$ edges is $O(m\log n\,\Delta(G)^{O(1)})$.

Now we turn to the setting of this work, where $D=\Delta(G)$ and $\Delta(G)\ge 2\mad(G)$.
For simplicity assume that $\mad(G)$ is bounded.
The problem we encounter is that in our case not all edges are weak. 
However, Sanders and Zhao~\cite{sanders-zhao-size} show that there is always at least one weak edge.
Their proof can be transformed to an algorithm that finds a weak edge $xy$, removes it, colors the graph recursively, and then extends the coloring.
Unfortunately, then the path $P_{xy}$ may be long.

In this work, we show that when average density is bounded, then at least a constant fraction of edges are weak, thus generalizing a similar lower bound of Chrobak and Nishizeki~\cite{chrobak-nishizeki} for bounded genus graphs. This allows for a recursive procedure as follows: find weak edges, color the other edges recursively, and then color the weak edges in time $O(m\log n\,\Delta(G)^{O(1)})$ using the approach of Sinnamon or Chrobak-Nishizeki. Since the number of weak edges is $\Omega(m)$, this is only $O(\log n\,\Delta(G)^{O(1)})$ per edge.

The last missing piece in this construction is getting rid of the dependency on $\Delta(G)$. The solution of this particular issue can be found by adapting a technique of Zhou, Nakano and Nishizeki~\cite{ZhouTreewidth}. Namely, one can show that edges of the input graph can be partitioned in such a way that each of the resulting induced graphs has maximum degree at least $2\mad(G)$ and at most $4\mad(G)+2$, and moreover the maximum degrees of all the individual graphs sum up to $\Delta(G)$, so we can just color each of them separately using disjoint palettes of colors, and merge the results. In this way, we trade the dependency on $\Delta(G)$ for the dependency on $\mad(G)$. Having $\Delta(G)$ bounded makes many issues easy, for example controlling the running time of the algorithm behind VAL.

\subsection{Organization of the paper}

This paper is organized as follows.
In Section~\ref{sec:lower} we show a lower bound for the number of weak edges for graphs with $\Delta\ge 2\ad(G)$.
At the end of the section we will also prove that the constant 2 in the assumption  cannot be improved if we want to get any positive lower bound.
Next, in Section~\ref{sec:efficient-val} we show an implementation of the Vizing Adjacency Lemma stating some of its properties that we need.
In Section~\ref{sec:main} we describe the complete algorithm and we prove Theorem~\ref{thm:main}. 
Finally, in Section~\ref{sec:conclusion} we summarize the paper and suggest some open problems.

\section{Lower bound for the number of weak edges}
\label{sec:lower}

By a $d$-vertex (resp. $d$-neighbor) we mean a vertex (resp. neighbor) of degree $d$.
For an integer $i$, let $d^i(v)$ be the number of $i$-neighbors of $v$.
Let $D$ be an integer.
In what follows, we call an edge $uv$ {\em $(D,u)$-weak} if $d^D(u)\le D - d(v) + [d(v)=D]$.
Moreover, $uv$ is $D$-weak if it is $(D,u)$-weak or $(D,v)$-weak.
Note that for $D\ge\Delta(G)+1$, every edge is $D$-weak.
For simplicity, whenever $D$ is known, we write that an edge is {\em $u$-weak} or {\em weak} when it is $(D,u)$-weak or weak, respectively.
An edge $uv$ is {\em strong} if it is not weak, i.e., it is neither $u$-weak nor $v$-weak.
A vertex is called {\em weak} if it is incident to at least one weak edge, otherwise it is called {\em strong}.

In this section we focus on the case $D=\Delta(G)$ and we show that for bounded average degree, if $\Delta(G)$ is large enough, then at least a constant fraction of edges are weak.

\begin{theorem}
    \label{thm:lower-bound-weak}
Let $G=(V,E)$ be a nonempty graph with no isolated vertices and such that $\Delta(G) \ge 2\ad(G)$.
Then the number of $\Delta(G)$-weak edges in $G$ is at least
\[\frac{|E|}{2\ad(G)^2}.\] 
\end{theorem}

\begin{proof}
    For brevity, we denote $\tilde{d}=\ad(G)$.
    We use the discharging method~\cite{CRANSTON2017766}, in a similar spirit as Sanders and Zhao~\cite{sanders-zhao-size}, with a crucial difference that here we need a linear lower bound, and not just an existence of a weak edge.
    We assign an initial {\em charge} to each vertex $v$ as follows:
    \[\ch(v) = d(v) - \tilde{d}.\]
    Note that the total charge is zero:
    \[\sum_{v\in V}\ch(v) = \sum_{v\in V}d(v) - |V|\sum_{v\in V}d(v)/|V| = 0.\]

    Next, we move the charge as follows: 
    every vertex $x$ of degree $d(x)\ge\tilde{d}+\tfrac12$, sends 
    \[\frac{\ch(x) -\tfrac12}{d(x)-\Delta+d(y)-1}\]
    units of charge to every neighbor $y$ of degree $<\tilde{d}+\tfrac12$ such that edge $xy$ is strong. 
    We claim that the expression above is always well defined and positive.
    Indeed, since $xy$ is strong, $d(x)\ge [d(y)<\Delta]+\dD(x)\ge[d(y)<\Delta]+\Delta-d(y)+[d(y)=\Delta]+1=\Delta-d(y)+2$. 
    This also implies that $d(y)\ge \Delta-d(x)+2\ge 2$, which we state below as a claim for later use.

    \begin{claim}
        \label{claim:deg2}
        If a vertex receives charge, then it has degree at least 2.
    \end{claim}

    For a vertex $v$, let $\ch'(v)$ denote the charge of $v$ after moving the charge. Since the charge is only being moved between vertices, $\sum_{v\in V}\ch'(v) = \sum_{v\in V}\ch(v) = 0$. Our goal is to show that
    \begin{enumerate}[$(i)$]
        \item for every strong vertex the final charge is at least $\tfrac12$, and
        \item for every weak vertex the final charge is at least $1-\tilde{d}$.
    \end{enumerate}
     Then, denoting the set of weak vertices by $V_w$, we have
    \[0 = \sum_{v\in V}\ch'(v) \ge |V_w|\cdot(1-\tilde{d})+(|V|-|V_w|)\cdot \tfrac12,\]
    which implies $(\tilde{d}-\tfrac12)|V_w|\ge |V|/2$. 
    Then, since $G$ has no isolated vertices, $\tilde{d}\ge 1$, and hence $|V_w|\ge |V|/(2(\tilde{d}-\tfrac12))> |V|/(2\tilde{d})$. Since every weak vertex is incident with at least one weak edge, we get at least $|V|/(4\tilde{d})=|E|/(2\tilde{d}^2)$ weak edges, as required.
    Hence, indeed it suffices to show $(i)$ and $(ii)$.

    \begin{claim}
        \label{claim:big-vertices}
    For any vertex $v$ of degree $d(v)\ge\tilde{d}+\tfrac12$ we have $\ch'(v)\ge \tfrac12$.
    \end{claim}

    \begin{claimproof}
        Since $\ch(v)\ge\tfrac12$, the claim holds trivially when $v$ does not send charge.
        Hence in what follows assume $v$ sends charge. 
        Then $v$ is incident to at least one strong edge.
        Let $w$ be a neighbor of $v$ such that $vw$ is strong and $d(w)$ is minimum among such neighbors.
        Note that $d(w)<\tilde{d}+\tfrac12$, for otherwise $v$ does not send charge to any vertex.
        Then $v$ sends at most
        \[\frac{d(v)-\tilde{d}-\tfrac12}{d(v)-\Delta+d(w)-1}\]
        to each neighbor $u$ of $v$ such that $vu$ is strong and $d(u)<\tilde{d}+\tfrac12$, and does not send to other vertices.
        Note that $v$ does not send to $\Delta$-vertices, because $\tilde{d}+\tfrac12\le \tfrac12\Delta + \tfrac12=\Delta-\tfrac12(\Delta-1)\le \Delta$.
        Since $vw$ is strong, $\dD(v)\ge \Delta - d(w) + 1$.
        Hence $v$ has at most $d(v)-\Delta + d(w) - 1$ neighbors of degree smaller than $\Delta$ and this is also an upper bound for the number of neighbors it sends charge to. It follows that in total $v$ sends at most $d(v)-\tilde{d}-\tfrac12$ charge, and hence it keeps at least $\tfrac12$ of charge, as required.
    \end{claimproof}        

    \begin{claim}
        \label{claim:small-vertices}
        For any strong vertex $v$ of degree $d(v)<\tilde{d}+\tfrac12$ we have $\ch'(v)\ge \tfrac32$.
        \end{claim}
    
        \begin{claimproof}
            Let $w$ be a neighbor of $v$ of minimum degree.
            Since $vw$ is strong, $d(w)\ge\dD(w) \ge \Delta-d(v)+1>2\tilde{d}-(\tilde{d}+\tfrac12)+1=\tilde{d}+\tfrac12$. It follows that all neighbors of $v$ send charge to $v$, and
            \[\ch'(v)=d(v)-\tilde{d}+\sum_{i=d(w)}^\Delta n_i\cdot \frac{i-\tilde{d}-\tfrac12}{i-\Delta+d(v)-1},\]
            where $n_i$ is the number of $i$-neighbors of $v$.
            Since $\Delta-d(v)+1 > 2\tilde{d} - (\tilde{d}+\tfrac12) + 1 = \tilde{d}+\tfrac12$, we infer that 
            $\frac{i-\tilde{d}-\tfrac12}{i-\Delta+d(v)-1}=1+\frac{\Delta-d(v)+1-\left(\tilde{d}+\tfrac12\right)}{i-\Delta+d(v)-1}$ is decreasing with increasing $i$, so
            \[\ch'(v)\ge d(v)-\tilde{d}+d(v)\cdot \frac{\Delta-\tilde{d}-\tfrac12}{d(v)-1}.\]
            Since $d(v)\ge 2$ by Claim~\ref{claim:deg2} and $\Delta\ge 2\tilde{d}$, we get $\ch'(v)\ge \tfrac32$, as required.
        \end{claimproof}        
    
    To finish the proof, note that $(i)$ follows from claims~\ref{claim:big-vertices} and~\ref{claim:small-vertices}.
    For $(ii)$, a weak vertex of degree at least $\tilde{d}+\tfrac12$ has final charge at least $\tfrac12$ by Claim~\ref{claim:big-vertices}, which is more than $1-\tilde{d}$ since $\tilde{d}\ge 1$ (no isolated vertices). Finally, a weak vertex $v$ of degree $<\tilde{d}+\tfrac12$ does not send charge, so it has final charge $\ch'(v)\ge\ch(v)=d(v)-\tilde{d}\ge 1-\tilde{d}$, as required.
    This ends the proof.
\end{proof}    

We conclude this section by showing that the assumption $\Delta(G)\ge 2\ad(G)$ in the formulation of Theorem~\ref{thm:lower-bound-weak} is the best possible.

\begin{theorem}
    \label{thm:best}
    For every $\epsilon>0$ there are arbitrarily large graphs $G$ with no isolated vertices, with no weak edges and such that $\Delta(G)\ge(2-\epsilon)\ad(G)$.
\end{theorem}

\begin{proof}
    Let $d$ be an integer such that $d\ge \max\{\tfrac3\epsilon-1,3\}$ and set $\Delta=d(d-1)$.
    Set also $n_d=(d-1)(\Delta-d+2)$ and $n_\Delta=d(\Delta-d+2)$.
    Construct graph $G$ as follows. 
    Begin with $\Delta-d+2$ copies of the complete bipartite graph $K_{d-1,d}$ (note that since $d\ge 2$, $\Delta-d+2=d(d-2)+2\ge 5$).
    Let $V_d$ be the set of the vertices of the resulting graph of degree $d$, and let $V_\Delta$ be the remaining vertices, i.e., the vertices of degree $d-1$.
    Note that $|V_d|=n_d$ and $|V_\Delta|=n_\Delta$.
    Next, we partition $V_\Delta$ arbitrarily into $d$ disjoint sets of size $(\Delta-d+2)$ and for each such set we add all edges between its verties so that it induces the clique $K_{\Delta-d+2}$. 
    Observe that then every vertex in $V_\Delta$ has degree $d-1+\Delta-d+1=\Delta$.
    Note that since $d\ge 2$, we have $\Delta\ge d$ and hence $G$ has maximum degree $\Delta$.
    This completes the construction of $G$.
    
    Then, we have
    \[\ad(G) = \frac{dn_d+\Delta n_\Delta}{n_d+n_\Delta} = \frac{\overbrace{d(d-1)}^\Delta+\Delta d}{d-1 + d} = \frac{d+1}{2d-1}\Delta.\]
    It follows that
    \[\Delta = \frac{2d-1}{d+1}\ad(G) = \left(2-\frac3{d+1}\right)\ad(G) \ge (2-\epsilon)\ad(G).\]

    Now it suffices to check that $G$ has no weak edges.
    Let $uv$ be an edge such that $u\in V_d$ and $v\in V_\Delta$.
    Then $uv$ is not $u$-weak because $\dD(u)=d-1>1=\Delta-d(v)+[d(v)=\Delta]$.
    Also, $uv$ is not $v$-weak because $\dD(v)=\Delta-d+1>\Delta-d=\Delta-d(u)+[d(u)=\Delta]$.
    Finally, consider an edge $uv$ between two vertices of $V_\Delta$.
    Then $\dD(u)=\Delta-d+1>1=\Delta-d(v)+[d(v)=\Delta]$, so by symmetry $uv$ is not weak.
\end{proof}    

\section{Algorithmic Vizing Adjacency Lemma}
\label{sec:efficient-val}

The goal of this section is to provide an efficient algorithm which implements the Vizing Adjacency Lemma (Theorem~\ref{th:VAL}).
The contents of this section is mostly standard and well-known.

Throughout this section by $D$ we denote the number of available colors.
Of course in this paper we are mostly interested in the case when $D$ is the maximum degree $\Delta$, but it is convenient to state results also for the case $D>\Delta$, since our algorithm colors subgraphs of the input graph recursively and at some point the maximum degree drops.
A {\em partial edge-coloring} of graph $G$ is a function $\pi:E(G)\rightarrow\{1,\ldots,D\}\cup\{\bot\}$ such that it is an edge-coloring when restricted to edges colored by $\{1,\ldots,D\}$.
For a vertex $x$ we define $\pi(x)$ as the set of colors used at edges incident with $x$, i.e., $\pi(x)=\{\pi(xy)\mid y \in N_G(x)\} \setminus \{\bot\}$.
Moreover, $\free(x)$ denotes the set of colors free at $x$, i.e., $\free(x)=\{1,\ldots,D\}\setminus \pi(x)$.
An $(a,b)$-{\em alternating path} is a path $P$ in $G$ such that its edges are colored alternately in $a$ and $b$. We call $P$ {\em maximal} if it cannot be extended to a longer $(a,b)$-alternating path. The unordered pair $\{a,b\}$ is called the {\em type} of $P$.

Fans, developed by Vizing~\cite{vizing:first} are useful structures in edge-colorings. The notion appears is several variations. Here we need a new one, which differs from the standard fan by modifying condition (F3) and adding condition (F4) below. The definition is somewhat inspired by multi-fans in the book of Stiebitz et al.~\cite{stiebitz}.

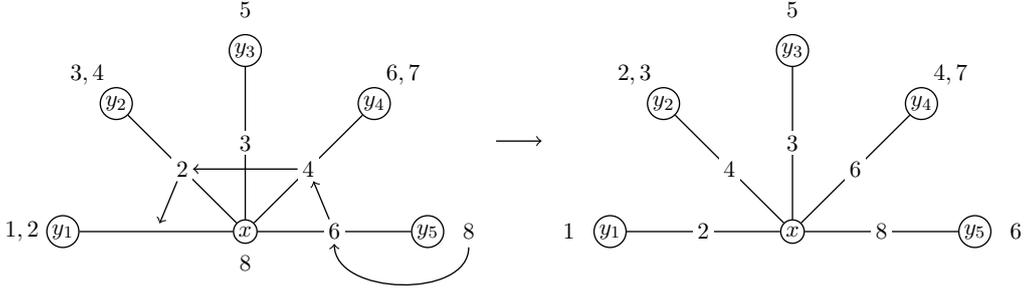
\begin{figure}
    \begin{center}
    \begin{tikzpicture}[line width=0.5pt,scale=0.6,every node/.style={scale=0.8}]
        \tikzstyle{vtx}=[draw,circle,fill=black,minimum size=5pt,inner sep=0pt]
        \tikzstyle{whitenode}=[draw,circle,fill=white,minimum size=11pt,inner sep=1pt]
        
        \draw (0,0) node[whitenode] (x) {$x$};
        \foreach \i/\c/\f in {1//{1,2},2/2/{3,4},3/3/5,4/4/{6,7},5/6/8}
        {    
            \draw (180-45*\i+45:4) node[whitenode] (y_\i) {$y_\i$};
            \draw (180-45*\i+45:4.9) node (z_\i) { $\f$};
            \ifthenelse{\NOT \i=1}{\draw (x)--(y_\i) node [midway, fill=white,inner sep=2pt] (c_\i){ $\c$};}{\draw (x)--(y_\i) node [midway] (c_\i){};}
        }    
    
        \draw [->] (c_5) -- (c_4);
        \draw [->] (c_4) -- (c_2);
        \draw [->] (c_2) -- (c_1);
        \draw [->] (z_5) to[out=-90,in=-90] (c_5);
        \draw (-90:0.7) node (z_0) { $8$};
        \draw [->] (5.5,2) -- (6.5,2);
    
        \begin{scope}[xshift=12cm]
            \draw (0,0) node[whitenode] (x) {$x$};
            \foreach \i/\c/\f in {1/2/{1},2/4/{2,3},3/3/5,4/6/{4,7},5/8/6}
            {    
                \draw (180-45*\i+45:4) node[whitenode] (y_\i) {$y_\i$};
                \draw (180-45*\i+45:4.9) node (z_\i) { $\f$};
                \draw (x)--(y_\i) node [midway, fill=white,inner sep=2pt] (c_\i){$\c$};
            }    
            \end{scope}        
    \end{tikzpicture}    
    \caption{\label{fig:fan}Rotating a fan. The numbers next to vertices denote free colors.}
\end{center}
\end{figure}

Let $G=(V,E)$ be a graph and let $\pi$ be a partial edge-coloring of $G$.
Then, a {\em fan} is a sequence $F=(x,y_1,\ldots,y_k)$ such that
\begin{enumerate}[(F1)]
    \item $y_1,\ldots,y_k$ are different neighbors of $x$,
    \item $xy_1$ is uncolored by $\pi$,
    \item for every $i=2,\ldots,k$ edge $xy_i$ is colored by $\pi$ and $\pi(xy_i)\in \bigcup_{j=1}^{i-1}\free(y_j)$,
    \item for every $i=2,\ldots,k$, $d(y_i)<D$.
\end{enumerate}
The number $k$ is called the {\em size} of fan $F$.
Since in this work we we do not use other variants of fan, we do not give it a special name. The notion of a fan is motivated by the fact that there is a simple procedure, called {\em rotating} the fan, which, given a color $c\in\free(x)\cap\free(y_k)$ recolors some edges of the form $xy_i$ so that edge $xy_1$ can be colored with a free color, see Algorithm~\ref{alg:rotate} and Figure~\ref{fig:fan}.
Of course, the definition of the fan does not guarantee that such a color $c$  exists.

\newcommand\mycommfont[1]{\footnotesize\ttfamily\textcolor{blue}{#1}}
\SetCommentSty{mycommfont}

\begin{algorithm}
    \SetKwInOut{Input}{input}
    \DontPrintSemicolon
    \caption{Rotating a fan.}\label{alg:rotate}
    \Input{a fan $F=(x,y_1,\ldots,y_k)$, color $c\in\free(x)\cap\free(y_k)$.}
    $i \gets k$\;
    \While{$i \ge 1$}{
        \tcp*[l]{Invariant: for some $j=1,\ldots,i$ we have $c\in\free(x)\cap\free(y_j)$}
        \If{$c \in \free(y_i)$}{
          $c' \gets \pi(xy_i)$\; 
          $\pi(xy_i) \gets c$\;
          $c \gets c'$ 
        }
        $i \gets i - 1$
    }
\end{algorithm}

\begin{lemma}
    Given graph $G$, its partial coloring $\pi$, a fan $F=(x,y_1,\ldots,y_k)$ and color $c\in\free(x)\cap\free(y_k)$, Algorithm~\ref{alg:rotate} finds a new partial coloring, which colors the same edges as $\pi$ and additionally $xy_1$.
\end{lemma}
\begin{proof}
    It suffices to observe that at the beginning of the loop the following invariant is satisfied: for some $j=1,\ldots,i$ we have $c\in\free(x)\cap\free(y_j)$. 
    %
\end{proof}

We call a fan $F=(x,y_1,\ldots,y_k)$ {\em active} when one of the two conditions below hold:
\begin{enumerate}[({A}1)]
    \item there is $i=1,\ldots,k$ such that $\free(y_i)\cap\free(x)\ne\emptyset$,
    \item there are $i,j=1,\ldots,k$ such that $i\ne j$ and  $\free(y_i)\cap\free(y_j)\ne\emptyset$.
\end{enumerate}
If additionally $F'=(x,y_1,\ldots,y_{k-1})$ is not active, we call $F$ a {\em minimal active fan}.

We will prove a series of simple lemmas. In all of them we work with a graph $G$ and its partial coloring $\pi$.

\begin{lemma}
    \label{lem:total-free}
    If a fan $F=(x,y_1,\ldots,y_k)$ is not active, then 
    \[\left|\bigcup_{i=1}^k \free(y_i)\right| \ge D - d(y_1)+k.\]
\end{lemma}

\begin{proof}
    Since $xy_1$ is uncolored, $|\pi(y_1)| \le d(y_1)-1$.
    Then, we have $|\free(y_1)| = D - |\pi(y_1)| \ge D - (d(y_1)-1)$ and by (F4), $|\free(y_i)| \ge 1$ for $i=2,\ldots,k$. It follows that
    \[\left|\bigcup_{i=1}^k \free(y_i)\right| \stackrel{(A2)}{=} \sum_{i=1}^k |\free(y_i)| \ge D - (d(y_1)-1) + (k-1)\cdot 1 = D - d(y_1)+k.\]
\end{proof}

\begin{lemma}
    \label{lem:extend}
    Let $xy_1$ be an $x$-weak edge.
    If a fan $F=(x,y_1,\ldots,y_k)$ is not active, then it can be extended to another fan $F'=(x,y_1,\ldots,y_{k+1})$. 
\end{lemma}

\begin{proof}
    Consider the set $Z=\{z\in N_G(x) \mid \pi(xz) \in \bigcup_{i=1}^k\free(y_i)\}$.
    Then, \[|Z| \stackrel{(A1)}{=} \left|\bigcup_{i=1}^k \free(y_i)\right| \stackrel{\text{Lemma \ref{lem:total-free}}}{\ge}D - d(y_1)+k. \]
    As $xy_1$ is $x$-weak, we have $d^D(x)\le D-d(y_1)+[d(y_1)=D]$.
    Since $y_1\not\in Z$, because $xy_1$ is not colored, we know that $Z$ contains at least $D - d(y_1)+k-(D-d(y_1))=k$ vertices of degree $<D$.
    It follows that $Z\setminus\{y_i \mid i=2,\ldots,k\}$ has at least one vertex of degree $<D$, call it $y_{k+1}$. It is easy to check that $(x,y_1,\ldots,y_{k+1})$ satisfies conditions (F1)-(F4).
\end{proof}    

Now we turn to lemmas that describe algorithms. These algorithms use the following data structures that represent the input graph and a partial coloring of it:
\begin{itemize}
    \item adjacency lists, i.e., array $N$ such that $N[v]$ is a list of neighbors of $v$;
    \item for every $v\in V$ we keep an associative array $M_v$ that maps every color $c\in\pi(v)$ to the edge incident with $v$ colored with $c$. The array is also used to test membership in $\pi(v)$.
    We use a simple array of size $D$. Note that the total size of these arrays is $O(nD)$, but we will use it only when $D$ is somewhat small, i.e., $D=O(\mad(G))$.
\end{itemize}

The lemma below is rather standard (similar statements, but for different variants of the fan notion, can be found for example in Sinnamon~\cite{sinnamon} or Bhattacharya~\cite{Bhattacharya23}).

\begin{lemma}
    \label{lem:fan-building-x}
    Given an uncolored $x$-weak edge $xy_1$, one can find a minimal active fan $F=(x,y_1,\ldots)$ in time $O(D)$. The algorithm also returns the color that makes (A1) or (A2) satisfied.
\end{lemma}

\begin{proof}
    Our algorithm is as follows. We begin with the fan $F_1=(x,y_1)$.
    Then we extend it vertex by vertex, until we get an active fan.
    Let us describe a single step of this process.
    We will use a set $S$ of colors, initially (i.e., before all steps) empty, implemented using a bitmap (of size $|D|$).
    We also use a queue $Q$ storing a subset of $S$.
    Let $F_i=(x,y_1,\ldots,y_i)$ be a fan built so far.
    We know that $i=1$ or $F_{i-1}$ is not active (otherwise $y_i$ would not be added).
    We maintain the invariant that $S$ consists of all the colors free at any $y_j$, $j<i$.
    First we check if $F_i$ is active.
    To this end, for each color $c\in\free(y_i)$ we check if $c\not\in\pi(x)$ or $c\in S$ in $O(1)$ time per color $c$. If one of these checks is positive (A1) or (A2) holds, respectively, so we stop (without checking further colors of $\free(y_i)$) and return $F_i$. Otherwise we add $c$ both to $S$ and $Q$ and proceed to the next $c$ (or finish the loop).
    Let us call the loop described above as the {\em activity checking loop}.
    When this loop finishes without returning a fan, we know that $F_i$ is not active.
    At this point we know by Lemma~\ref{lem:extend} that $F_i$ can be extended by a vertex $y_{i+1}$.
    Our algorithm finds $y_{i+1}$ as follows.
    We remove a color $c$ from $Q$, and find the edge $xz$ colored with $c$ using $M_x$. If $d(z)<D$ we set $z=y_{i+1}$, and otherwise we remove another color from $Q$ and so on.
    This loop will be called the {\em extension loop}.
    
    Now let us analyze the time complexity. Let $F=(x,y_1,\ldots,y_k)$ be the fan returned.
    Clearly, each iteration of each of the two loops takes constant time, so it suffices to bound the total number of iterations.
    For the activity checking loop, note that by (A1) each but the last iteration corresponds to a different color that is in $\pi(x)$ (it is not the case for the last iteration only if $F$ was returned at this iteration). Hence the number of these iterations is at most $|\pi(x)|\le d(x)-1$.
    Finally, the number of iterations in the extension loop is bounded by the number of colors inserted to $Q$. However, we insert to $Q$ only in the activity checking loop, one element per iteration.
    Hence the total number of insertions to $Q$ is bounded by $d(x)-1$.
    It follows that the total number of iterations in both loops is $O(d(x))$, and each of them is done in constant time.
    Since $d(x)\le D$ and initializing bitmap $S$ takes $O(D)$ time, we get the $O(D)$ time bound as required. 
    \end{proof}    
    
Now we are ready to prove the main theorem of this section. 
Its proof is basically an algorithmization of the proof of Theorem 2.1 from Stiebitz et al.~\cite{stiebitz}, using Lemma~\ref{lem:fan-building-x}. Recall that for a vertex $v$ the {\em closed neighborhood} of $v$ is the set $N[v]=N(v)\cup\{v\}$.

\begin{theorem}[Algorithmic VAL]
    \label{th:weak-edge-coloring}
    Let $G$ be a simple graph and let $e=xy$ be an $x$-weak edge.
    Then there is an algorithm which, given any partial $D$-edge-coloring $\pi$ of $G$ which colors a subset $E_c$ of edges of $G$, $e\not\in E_c$ finds a partial $D$-edge-coloring $\pi'$ that colors $E_c\cup\{e\}$.
    The algorithm runs in time $O(D+|P|)$, where $P$ is a maximal alternating path with an endpoint in $x$, possibly $P=\emptyset$. Moreover,
    \begin{enumerate}[$(V1)$]
        \item if needed, one can specify the color $c'\in\free(x)$ that alternates in $P$,
        \item $\pi'$ may differ from $\pi$ only on edges incident to $x$ and on edges of $P$,
        \item if $\pi(v)\ne\pi'(v)$ for a vertex $v$, then $v\in N[x]$ or $v$ is the other endpoint of $P$,
        \item the type of $P$ can be determined in time $O(D)$,
        \item the type of $P$ depends only on $\{(\pi(xz),\pi(z))\mid z \in N(x)\}$.
    \end{enumerate}
\end{theorem}

\begin{proof}
    Let $y_1=y$. We apply Lemma~\ref{lem:fan-building-x}.
    Let $F=(x,y_1,\ldots,y_k)$ be the returned minimal active fan.
    Note that $k\le d(x) \le D$.

    Since $F$ is active, it satisfies (A1) or (A2). 
    In the former case, we have a color $c\in\free(y_k)\cap\free(x)$.
    Then we {\em rotate $F$} using Algorithm~\ref{alg:rotate}, which clearly runs in time $O(k)$ since the condition in line 3 is checked in  constant time using $M_{y_1}$.
    By property (F3), this results in the desired partial edge-coloring.
    Note that $(V2)$ is satisfied, and $(V1)$, $(V3)$-$(V5)$ are trivial since $P=\emptyset$.

    Now assume $F$ does not satisfy (A1) but it does satisfy (A2), i.e., by Lemma~\ref{lem:fan-building-x} we have a color $c$ that is free both in $y_k$ and $y_i$ for some $i=1,\ldots,k-1$. Such $y_i$ can be found simply by testing if $c\in\pi(y_j)$ for every $j=1,\ldots,k-1$: this takes time $O(k)=O(D)$.

    Let $c'$ be any color free at $x$ (it exists since $xy$ is uncolored, and if $c'$ is not given, it can be found in time $O(D)$).
    Find a maximal $(c,c')$-alternating path $P$ starting at $x$. This can be easily done in time $O(|P|)$ using the associative arrays $M_{*}$.
    Then there is $j\in\{i,k\}$ such that the other endpoint of the path $P$ is not $y_j$. 
    We swap the colors $c$ and $c'$ on $P$ in $O(|P|)$ time, obtaining another partial coloring with the same set of colored edges, but with $c\in\free(x)\cap\free(y_j)$. Then we rotate $F_j$ using Algorithm~\ref{alg:rotate}. We see that this procedure satisfies $(V2)$, $(V3)$ and $(V5)$. For $(V4)$, we just perform the procedure without finding $P$ and without recoloring.
\end{proof}

\section{The main algorithm}
\label{sec:main}

In this section we present our main result.

We begin with two lemmas which will be used to control the length of alternating paths from Theorem~\ref{th:weak-edge-coloring}.
The first one is due to Sinnamon~\cite{sinnamon} and uses randomization.

\begin{lemma}[Lemma 10 in Sinnamon~\cite{sinnamon}]
    \label{lem:expected-length}
    Let $G$ be a simple graph and with a partial edge-coloring $\pi$ which colors all but $\ell$ edges of $G$.
    Let $e$ be a random uncolored edge, let $x$ be an arbitrary endpoint of $e$ and let $c_0$ be a random color of $\pi(x)$.
    Let $c_1$ be an arbitrary color different than $c_0$.
    Consider the unique maximal $(c_0,c_1)$-alternating path $P_r$ starting at $x$.
    Then, $\Expect{|P_r|} = O(|E(G)|\Delta(G)/\ell)$.    
\end{lemma}

Note that if we iterate Theorem~\ref{th:weak-edge-coloring} for $\ell=\Omega(|E|)$ times, then by Lemma~\ref{lem:expected-length} the expected length per edge is low, namely $O(\Delta(G)H_\ell)=O(\Delta(G)\log n)$.
Below we state a lemma of Chrobak and Nishizeki, which achieves a similar goal deterministically.
Since their proof does not analyze the dependency on $\Delta(G)$, we reproduce it here, but involving $\Delta(G)$ in the analysis.

\begin{lemma}[Chrobak and Nishizeki~\cite{chrobak-nishizeki}]
    \label{lem:chrobak-nishizeki}
    Let $G$ be a simple graph and let $E_w\subseteq E(G)$ be a subset of weak edges. Denote $E_0=E(G)\setminus E_w$. 
    There is a deterministic algorithm which, given $G$, $E_w$ and a partial $\Delta$-edge-coloring $\pi$ which colors $E_0$, finds a new partial $\Delta$-edge-coloring $\pi'$ of $G$ which colors $E_1$ such that $E_0\subseteq E_1$ and $|E_1\setminus E_0| \ge |E_w|/(9\Delta(G)^5)$. The algorithm runs in time $O(|E(G)|+|E_w|\Delta(G))$.
\end{lemma}

\begin{proof}
First, for every edge $xy\in E_w$ we determine the type $\{a,b\}$ of a maximal alternating path $P_{xy}$ that would be used if $xy$ was colored using Theorem~\ref{th:weak-edge-coloring}; let's call it the type of $xy$.
For the special case $P_{xy}=\emptyset$, the type of $xy$ is also defined as $\emptyset$.
By Theorem~\ref{th:weak-edge-coloring}, property $(V4)$, this takes $O(|E_w|\Delta(G))$ time.

Since there are ${\Delta(G)\choose 2}+1 \le \Delta^2$ types, there is a type, call it $t$, with at least $|E_w|/\Delta(G)^2$ edges. Let $E_w^t$ be a set of $\lceil |E_w|/\Delta(G)^2\rceil$ edges of type $t$.

Consider two two different edges $e=xy,e'=x'y'$, both from $E_w^{a,b}$, that are $x$-weak and $x'$-weak, respectively.
We say that $e$ {\em interacts} with $e'$ if coloring $e$ using Theorem~\ref{th:weak-edge-coloring} changes the type of $e'$ or recolors an edge of $P_{e'}$, or vice versa, i.e., coloring $e'$ changes the type of $e$ or recolors an edge of $P_{e}$.

\begin{claim}
	Edges $e$ and $e'$ do not interact when
	\begin{enumerate}[$(i)$]
		\item $N[x]\cap N[x']=\emptyset$, and
		\item $P_e$ does not end in $N[x']$ and $P_{e'}$ does not end in $N[x]$.
	\end{enumerate}
\end{claim}
\begin{claimproof}
	Consider $z'\in N(x')$.
	Assume $\pi(x'z')$ changes after coloring $e$ using Theorem~\ref{th:weak-edge-coloring}.
	By $(V2)$, $x'z'$ is incident to $x$ or $x'z'\in P_e$. 
	The former is excluded by $(i)$, so assume the latter.
	It means $t=\{a,b\}$ for some colors $a,b$. 
	By $(ii)$, path $P_e$ does not end in $x'$, so $a$ and $b$ are not free at $x'$, a contradiction with $P_{e'}$ being a maximal $(a,b)$-path with an endpoint in $x'$.
	Hence $\pi(x'z')$ does not change.
	Similarly, $\pi(z')$ does not change, for otherwise by $(V3)$ we have $z'\in N[x]$, contradicting $(i)$, or $P_e$ ends in $z'$, contradicting $(ii)$.
	Thus, by $(V5)$, the type of $e'$ does not change.
	
	Assume an edge of $P_{e'}$ was recolored after coloring $e$ using Theorem~\ref{th:weak-edge-coloring}.
	Then $P_{e'}\ne\emptyset$, so $t=\{a,b\}$ for some colors $a,b$.
	By $(ii)$, we have $P_{e'}\ne P_e$, so by $(V2)$ path $P_{e'}$ has an edge incident to $x$.
	Also by $(ii)$, path $P_{e'}$ does not end at $x$, contradicting $P_e$ being a maximal $(a,b)$-path with an endpoint in $x$.
\end{claimproof}
	
We proceed as follows. 
We form a graph $Q=(E_w^t, E(Q))$, where $E(Q)$ is the set of pairs $e,e'\in E_w^t$ that violate $(i)$. By definition, $\Delta(Q)\le \Delta(G)(\Delta(G)-1)^2$ so a simple greedy algorithm finds a maximal independent set $S$ in $Q$ of size 
\[|S|\ge \frac{|E_w^t|}{\Delta(Q)+1}\ge \frac{|E_w^t|}{\Delta(G)^3}\ge \frac{|E_w|}{\Delta(G)^5}.\]
The greedy algorithm runs in time linear in the size of $Q$, which is $|E_w^t|\Delta(Q)=O(E_w \Delta(G))$.

Finally, when $t\ne\emptyset$, i.e., $t=\{a,b\}$, we form yet another graph $Q'=(S, E(Q'))$, where $E(Q')$ is the set of pairs $e,e'\in S$ that violate $(ii)$. 
To form $Q'$ we need to find all the $(a,b)$-alternating paths that correspond to edges in $S$, but it takes only time $O(|E(G))$ thanks to the fact that these paths are pairwise disjoint. 
We claim that $|E(Q')|\le |V(Q')|$. 
Indeed, consider $xy\in S$ and the $(a,b)$-alternating path $P_{xy}$ starting at $x$.
Let $w$ be the other endpoint of $P_{xy}$. Since edges of $S$ satisfy $(i)$, $x$ is not in a closed neighborhood of an endpoint of an edge from $S$ different from $e$.
If $w$ is in a closed neighborhood of an endpoint of an edge $e_w$ from $S$, then again by $(i)$ we know that $e_w$ is unique. Hence, $xy$ generates at most one edge in $Q'$, and $|E(Q')|\le |V(Q')|$. It follows that $Q'$ has a set $Z$ of at least $|V(Q')|/3$ vertices of degree at most 2, because otherwise $|E(Q')|=\tfrac12\sum_{e\in V(Q')}d_{Q'}(e)>\tfrac12 \cdot 3\cdot\tfrac23|V(Q')|=|V(Q')|$, a contradiction.
By repeatedly selecting a vertex of $Z$ and discarding its at most two neighbors we get an independent set $I$ in $Q$ of size 
\[|I|\ge|V(Q')|/9=\frac{|E_w|}{9\Delta(G)^5}.\]
This takes linear time in the size of $Q'$, which is $O(|S|)=O(E_w)$.
If $t=\emptyset$ we just set $I=S$.

We conclude by applying Theorem~\ref{th:weak-edge-coloring} to color all edges of $I$. Since the corresponding alternating paths are disjoint, and they do not change because edges of $I$ do not interact with each other, their total length is bounded by $|E(G)|$. 
By Theorem~\ref{th:weak-edge-coloring}, it follows that the total time of this coloring is $O(|I|\Delta(G) +|E(G)|)=O(|E_w|\Delta(G)+|E(G)|)$.
\end{proof}

Now we need a graph partitioning technique due to Zhou, Nakano, Nishizeki~\cite{ZhouTreewidth}. Since we state the properties of this technique somewhat differently than them, below we provide a self-contained proof.

Let $G=(V,E)$ be a graph.
For an integer $c$, a {\em $(\Delta,c)$-partition} of $G$ is any partition $E=E_1\cup E_2 \cup\cdots E_k$ such that for graphs $G_i=G[E_i]$, $i=1,\ldots,k$ we have
\begin{enumerate}[$(P1)$]
    \item $\Delta(G) = \sum_{i=1}^k \Delta(G_i)$, and
    \item $\Delta(G_i) = c$ for $i=1,\ldots,k-1$, while $c\le\Delta(G_k)<2c$.
\end{enumerate}

Recall that the {\em degeneracy} of $G$ is the smallest integer $k$ such that every subgraph of $G$ has a vertex of degree at most $k$.

\begin{lemma}[implicit in Zhou, Nakano, Nishizeki~\cite{ZhouTreewidth}]
    \label{lem:partition}
    For every graph $G$ of degeneracy $k$, for every integer $c\ge k$, a $(\Delta,c)$-partition of $G$ can be found in linear time.
\end{lemma}

\begin{proof}
    Set $s=\lfloor \frac{\Delta(G)}{c} \rfloor$. It follows that $sc\le\Delta(G)<(s+1)c$.

    Let $n=|V(G)|$ and let $v_1,\ldots,v_n$ be the degeneracy ordering of $V(G)$, i.e., for every $i=1,\ldots,n$, there are at most $k$ edges between $v_i$ and vertices $v_1,\ldots,v_{i-1}$. It is well known that degeneracy ordering can be found in linear time (by repeatedly removing vertices of the smallest degree and reversing the order).

    Sets $E_1,\ldots,E_s$ are constructed using the following procedure.
    Begin with $E_i=\emptyset$ for $i=1,\ldots,s$.
    Next process vertices $v_i$, $i=1,\ldots,n$ one by one, as follows.
    For every vertex $v_i$ we assign the edges incident to $v_i$ and vertices $\{v_j \mid j>i\}$ to sets $E_1,\ldots,E_s$ in a way we describe later. This means that when we process $v_i$, its incident edges to vertices $\{v_j \mid j<i\}$ are already assigned. 
    However, since we use the degeneracy ordering, the total number of these latter edges is at most $k$, so in particular for every $\ell=1,\ldots,s$ the set $E_\ell$ contains at most $k$ edges incident to vertices $\{v_j \mid j<i\}$. Then, for every $\ell=1,\ldots,s-1$ we proceed as follows. 
    Let $r$ be the number of the yet not assigned edges from $v_i$ to $\{v_j \mid j>i\}$. Assign $\min\{c-d_{G[E_\ell]}(v_i),r\}$ edges to $G[E_\ell]$. If all edges incident to $v_i$ are assigned, proceed to $v_{i+1}$.
    If after adding edges to all graphs $E_\ell$, $\ell=1,\ldots,s-1$ there are still unassigned edges left, assign all of them to $E_s$. 
    
    Note that after processing any $v_i$, we have $d_{G[E_\ell]}(v_i)\le c$ for every $\ell=1,\ldots,s-1$. Also, $d_{G[E_s]}(v_i)< 2c$, because otherwise $d_G(v_i)=\sum_{\ell=1}^s d_{G[E_\ell]}(v_i) \ge (s+1)c>\Delta(G)$, a contradiction. 
    These degrees do not change when we process vertices $\{v_j \mid j>i\}$.
    Hence, $\Delta(G[E_\ell])\le c$ for every $\ell=1,\ldots,s-1$ and $\Delta({G[E_s]})< 2c$. 
    
    Now consider $v_i$ such that $d_G(v_i)=\Delta(G)$.
    We claim that $d_{G[E_\ell]}(v_i)= c$ for every $\ell=1,\ldots,s-1$.
    Indeed, otherwise at most $c(s-1)-1$  edges incident to $v_i$ were assigned to all $E_\ell$, $\ell=1,\ldots,s-1$ and at most $k$ edges incident to $v_i$ were assigned to $E_s$, and these are all edges incident to $v_i$ so $\Delta(G)=d_G(v_i)\le c(s-1)-1+k \le cs-1$, a contradiction with the choice of $s$. Hence indeed $d_{G[E_\ell]}(v_i)= c$ for every $\ell=1,\ldots,s-1$. This, together with $sc\le d_G(v_i)<(s+1)c$ implies that $c\le d_{G[E_s]}(v_i)<2c$. By the previous paragraph we get that $\Delta({G[E_\ell]})= c$ for every $\ell=1,\ldots,s-1$ and $c\le \Delta({G[E_s]})<2c$, as required.
\end{proof}

Now we proceed with a lemma that is a basic building block of our main result.

\begin{lemma}
    \label{lem:main}
    Every graph $G$ with $n$ vertices and $m$ edges such that $\Delta(G)\ge 2\mad(G)$ can be $\Delta(G)$-edge-colored 
    \begin{enumerate}
        \item by a randomized algorithm running in time $O(m\Delta(G)\mad(G)^2\log n)$ in the expectation and with high probability,
        \item by a deterministic algorithm in time $O(m\Delta(G)^5(\mad(G)^2+\Delta(G))\log n)$.
    \end{enumerate}
\end{lemma}

\begin{proof}
    The algorithm is recursive: it invokes itself recursively for a subgraph of the given graph.
    Consider a single recursive call and let $H\subseteq G$ be the graph passed to the current call.
    We have $D=\Delta(G)$ colors.
    We begin by removing all isolated vertices from $H$.
    If we end up with the empty graph, we return the empty coloring.
    Otherwise, we find the set $E_w$ of weak edges of $H$.
    This can be easily done in linear time, by first computing, for each vertex $v$ its degree $d(v)$ and $\Delta$-degree $\dD(v)$; then checking if an edge is weak directly from definition takes $O(1)$ time.
    By Theorem~\ref{thm:lower-bound-weak}, $|E_w|\ge |E(H)|/(2\mad(G)^2)$.
    In what follows we assume     $|E_w| = \lceil |E(H)|/(2\mad(G)^2)\rceil$ by skipping some edges if needed.
    Next, we form a graph $H'=H-E_w$ in time linear in the size of $H$.
    Then, if $E(H')\ne \emptyset$, we color $H'$ recursively. 
    Thus we obtain a partial coloring $\pi$ of $H$.

    Now we proceed differently depending on whether we aim at randomized or deterministic algorithm.
    In the randomized case, we color edges of $E_w$ in random order, using Theorem~\ref{th:weak-edge-coloring} (the color $c'$ mentioned in the theorem statement is chosen randomly: it takes time $O(D)$ to generate the list of free colors and next we sample from it in constant time).
    In the deterministic case, we apply Lemma~\ref{lem:chrobak-nishizeki} repeatedly until everything is colored (in every iteration $E_w$ in Lemma~\ref{lem:chrobak-nishizeki} is the set of yet uncolored edges).
    This completes the description of the algorithm.

    Now we proceed to the running time analysis.    
    Let us analyze the time complexity of the single level of the recursion.

    Begin with the randomized case.
    By Theorem~\ref{th:weak-edge-coloring} and Lemma~\ref{lem:expected-length}, when $\ell=1,\ldots,|E_w|$ edges are still uncolored, coloring a random uncolored edge $xy$ of $E_w$ takes expected time $O\left(\Delta(G)+|E(H)|\Delta(G)/\ell\right)=O(|E(H)|\Delta(G)/\ell)$. Hence, coloring all the $\ell$ edges is done in expected time $O(|E(H)|H_{|E_w|}\Delta(G))$, where $H_{|E_w|}$ is the $|E_w|$-th harmonic number. Since $H_{|E_w|}<\ln |E_w|+1 = O(\log n)$ we get the bound $O(|E(H)|\Delta(G)\log n)$ and this is also a bound on the expected running time for a single level of the recursion, since $E_w$ and $H'$ can be found in time $O(|E(H)|)$.

    For the deterministic case, we see that after each invocation of Lemma~\ref{lem:chrobak-nishizeki}, the number of uncolored edges decreases at least by the factor $r_0=1-\frac{1}{9\Delta(G)^5}$.
    The number of invocations is then at most $\log_{1/r_0}|E_w|=O(\Delta(G)^5\log|E_w|)$ by standard properties of logarithm.
    By Lemma~\ref{lem:chrobak-nishizeki}, it follows that in the deterministic case a single level of the recursion takes time $O((|E(H)|+|E_w|\Delta(G))\Delta(G)^5\log n)=O((1+\Delta(G)/\mad(G)^2)|E(H)|\Delta(G)^5\log n)$.

    Now we bound the time of all the levels of the recursion.
    Note that if $\Delta(H)<\Delta(G)=D$, then every edge of $H$ is weak and there are no more recursive calls.
    Otherwise, we have $D=\Delta(H)$ and $\Delta(H)=\Delta(G)\ge2 \mad(G)\ge 2\ad(H)$, so
    by Theorem~\ref{thm:lower-bound-weak} we have $|E(H')| \le r|E(H)|$ for
    \[r=1-\frac{1}{2\ad(H)^2}\le 1-\frac{1}{2\mad(G)^2}.\]
    Hence, in the randomized case, all the $L$ levels of the recursion take expected time
    \[O\left(\sum_{i=0}^{L-1}r^i|E(G)|\Delta(G)\log n\right).\]
    Since $\sum_{i=0}^{L-1}r^i<\sum_{i=0}^{\infty}r^i=2\mad(G)^2$, 
    the total time complexity is $O(m\Delta(G)\log n\mad(G)^2)$ in expectation, as required. 
    The same kind of calculation for the deterministic case gives the running time bound of  $O(m(\mad(G)^2+\Delta(G))\Delta(G)^5\log n)$.

    Obtaining the high probability guarantee for the running time in the randomized case is fairly standard, see Appendix~\ref{sec:whp}.
\end{proof}

Now we are ready to prove our main result.

\ourmainthm*

\begin{proof}
    Let $c=\lceil 2\mad(G)\rceil$.
    We begin with finding a $(\Delta,c)$-partition $E_1,\ldots,E_s$ of $G$ using Lemma~\ref{lem:partition}. Note that by definition, degeneracy is at most $\mad(G)$ so the assumptions of the lemma are satisfied.
        
    Thanks to property (P2) of $(\Delta,c)$-partition, for every  $i=1,\ldots,s$,
    \[2\mad(G[E_i])\le 2\mad(G)\le\Delta(G[E_i])< 2\lceil 2\mad(G)\rceil<4\mad(G)+2.\]

    Hence we can use Lemma~\ref{lem:main} to $\Delta(G[E_i])$-edge-color each graph $G[E_i]$, $i=1,\ldots,s$ separately, in expected time $O(|E_i|\log n\mad(G)^3)$ or deterministic time $O(|E_i|\mad(G)^7\log n)$. 
    In the randomized case, by linearity of expectation, we get $O(m\log n\mad(G)^3)$ expected time in total (and the high probability bound also follows since the probability of a bad event increases only at most $s\le\Delta\le n$ times).
    In the deterministic case, after summing we get $O(m\mad(G)^7\log n)$. 
    Since by the property (P1) we can use disjoint sets of colors, this results in the desired coloring of $G$.
\end{proof}

\section{Conclusion and further research}
\label{sec:conclusion}

We showed that every graph $G$ of bounded maximum average degree can be edge-colored using the optimal number of $\Delta$ colors in quasilinear time, whenever $\Delta\ge 2\mad(G)$. We presented two algorithms: a randomized one in time $O(m\mad(G)^3\log n)$ and a deterministic one in time $O(m\mad(G)^7\log n)$. As a corollary we get also that every graph $G$ of bounded maximum average degree can be $(\Delta+1)$-edge-colored in quasilinear {\em deterministic} time, thus derandomizing the result of Bhattacharya et al~\cite{Bhattacharya23}.
We conclude with some open problems as follows.
\begin{enumerate}
	
	\item Is it possible to relax our assumption to $\Delta\ge c\mad(G)$ for $c<2$? As shown in Theorem~\ref{thm:best} this would involve something more than the VAL in the coloring procedure.
	Since the graphs constructed in Theorem~\ref{thm:best} have arbitrarily large minimum degree, it is also not sufficient to find reducible configurations with vertices of bounded degree, as it was the case in the work of Cole and Kowalik~\cite{cole-kowalik}.
	\item Is it possible to design a data structure for {\em dynamic} $\Delta$-edge-coloring graphs of bounded maximum average degree, assuming $\Delta\ge c\mad(G)$ for a constant $c$, with updates in sublinear time? Or in polylogarithmic time? Note that the dynamic edge-coloring is an active area of research recently, see e.g.~\cite{DuanHZ19,AChristiansen-sparse-dynamic,Bhattacharya-sparse-dynamic,ABGChristiansen-power}.
	\item Is it possible to implement a version of Theorem~\ref{thm:main} in the LOCAL model using only ${\rm poly}(\Delta,\log n)$ rounds? For some recent related results see e.g.~\cite{ABGChristiansen-power,BernshteynVizingJCTB}.
\end{enumerate}

\section*{Acknowledgements}
The author wishes to thank Bartłomiej Bosek, Jadwiga Czyżewska, Konrad Majewski and Anna Zych-Pawlewicz for helpful discussions on related topics.
The author is also grateful to the reviewers for helpful comments.

\bibliographystyle{plainurl}
\bibliography{sparse-edge-col}

\newpage
\appendix
\section{High probability guarantee for the running time}
\label{sec:whp}

Recall that in our randomized algorithm in Lemma~\ref{lem:main}, randomization is used only when we apply Lemma~\ref{lem:expected-length}, which bounds the expected length of a maximal alternating path that is found and swapped by $O(|E(H)|\Delta(G)/\ell)$, where $\ell$ is the number of uncolored edges in $G$. Call the corresponding random variable $X_\ell$, $\ell=1,\ldots,|E_w|$. Let $X=\sum_{\ell=1}^{|E_w|}X_\ell$. Our goal is to show that w.h.p. $X$ is not much larger than its expectation, i.e., $O(|E(H)|\Delta(G)H_{|E(H)})$.

First consider the case $|V(H)|\le 4 n^{1/2} \ad(H)$.
Then, at this and all deeper levels of the recursion every maximial alternating path has length at most $|V(H)|-1$. Since the number of levels is at most $|E(H)|=O(n^{1/2}\ad(H)^2)$, we get that total length of maximal alternating paths on this and deeper levels is $O(n\ad(H)^3)$ with probability 1.

Now assume $|V(H)|\ge 4 n^{1/2} \ad(H)$.
We proceed as Bernshteyn and Dhawan~\cite{BernshteynDhawanBoundedDelta}, Section 4. 
We will need a version of multiplicative Azuma's Inequality, as follows.

\begin{theorem}[Corollary 6 in Kuszmaul and Qi~\cite{KuszmaulQi}]
    \label{th:azuma}
    Let $X_1, \ldots, X_n \in [0, c]$ be real-valued random variables with $c>0$. Suppose that $\Expect{X_i \mid X_1, \ldots, X_{i - 1}} \le a_i$ for all $i$. Let  $\mu=\sum_{i=1}^n a_i$. Then for any $\delta >0$, 
\[\Pr{\sum_i X_i\geq (1 + \delta)\mu}\leq \exp\left(-\frac{\delta^2\mu}{(2+\delta)c}\right).\]
\end{theorem}

Note that in fact for $i=1,\ldots,|E_w|$ we have $\Expect{X_i|X_1,\ldots,X_{i-1}} \le a_i$ for $a_i=\Theta(1 +\tfrac1i|E(H)|\Delta(G))$ and $X_i$ takes values in $[0,|V(H)|]$. Set $\mu = \sum_{i=1}^{|E_w|} a_i = \Theta(H_{|E_w|}\cdot |E(H)|\Delta(G))=\Theta((\log{|E_w|})\cdot |E(H)|\Delta(G))$.
By Theorem~\ref{th:azuma}, we get that for any $\delta \ge 1$,
\[\Pr{\sum_{i=1}^{|E_w|} X_i \ge (1+\delta)\mu} \le \exp\left(-\frac{\delta^2\mu}{(2+\delta) |V(H)|}\right) \le \exp(-\Omega(\delta\log{|E_w|}))=\frac1{{|E_w|}^{\Omega(\delta)}},\]
where in the second inequality we used the facts that $\delta\ge 1$ (so $\frac\delta{2+\delta}\ge\tfrac12$) and $H$ has no isolated vertices (so $|E(H)|\ge |V(H)|/2$).
Since by Theorem~\ref{thm:lower-bound-weak}, ${|E_w|} \ge |E(H)|/(2\ad(H)^2) = |V(H)|/(4\ad(H)) \ge n^{1/2}$, we get that 
$\Pr{\sum_{i=1}^{|E_w|} X_i \ge (1+\delta)\mu} \le n^{-\Omega(\delta)}.$ It follows that with high probability the sum of all maximal alternating path lengths that were found and swapped is $O(\mu)$, as required.

\end{document}